\documentclass[runningheads]{llncs}

\usepackage[linesnumbered,ruled,vlined]{algorithm2e}
\usepackage{cite}
\usepackage[T1]{fontenc}
\usepackage{libertine}
\usepackage[normalem]{ulem}
\usepackage{graphicx}
\newcommand{\remove}[1]{}
\usepackage{nth}
\usepackage{xcolor}

\usepackage{booktabs}
\usepackage{mathtools}

\usepackage{multicol}

\usepackage{mathtools}
\DeclarePairedDelimiter\ceil{\lceil}{\rceil}

\usepackage{float}
\usepackage{appendix}
\newenvironment{sloppypar*}
 {\sloppy\ignorespaces}
 {\par}
\usepackage{amssymb}

\begin{document}
\title{ 
Polynomial Time $k$-Shortest Multi-Criteria Prioritized and All-Criteria-Disjoint Paths}
\titlerunning{Multi-Criteria Shortest Paths}
%
\author{Yefim Dinitz\inst{} \and
Shlomi Dolev\inst{} \and
Manish Kumar\inst{}}

\authorrunning{Y. Dinitz et al.}
%

\institute{
Ben-Gurion University of the Negev, Be’er Sheva, Israel\\
\email{\{dinitz@cs, dolev@cs, manishk@post\}.bgu.ac.il}
}

\maketitle              
%
\begin{abstract}
\begin{sloppypar*}
The shortest secure path (routing) problem in communication networks has to deal with multiple attack layers e.g., man-in-the-middle, eavesdropping,  packet injection, packet insertion, etc. Consider different probabilities for each such attack over an edge, probabilities that can differ across edges. Furthermore, a usage of a single shortest paths (for routing) implies possible traffic bottleneck, which should be avoided if possible, which we term {\em pathneck security avoidance}. Finding all Pareto–optimal solutions for the multi-criteria single-source single-destination shortest secure path problem with non-negative edge lengths might yield a solution with an exponential number of paths. In the first part of this paper, we study specific settings of the multi-criteria shortest secure path problem, which are based on prioritized multi-criteria and on $k$-shortest secure paths. In the second part, we show a polynomial-time algorithm that, given an undirected graph $G$ and a pair of vertices $(s,t)$, finds prioritized multi-criteria $2$-disjoint (vertex/edge) shortest secure paths between $s$ and $t$. In the third part of the paper, we introduce the $k$-disjoint all-criteria-shortest secure paths problem, which is solved in time $O(\min(k|E|, |E|^{3/2}))$.
\end{sloppypar*}
\end{abstract}
\section{Introduction and Related Work}
\label{sec:introduction}
The shortest secure path (routing) problem in communication networks has to deal with multiple attack layers e.g., man-in-the-middle, eavesdropping,  packet injection, packet insertion, etc. Consider different probabilities for each such attack over an edge that can  differ across edges. We consider these attacks on the edge as multi-criteria and these criteria (attacks) have some positive weights. We need to compute the shortest secure path with the least value of total weight. We study a generalization of the shortest path problem in which multiple paths should be computed with consideration to multiple criteria. In our scenario, these criteria values can be proportional to the attack probabilities on the edges. The application of these paths is mainly in delivering messages via most secure routes. Similar applications are in the area of transportation networks, social networks, etc.

The shortest secure path routing algorithms mainly compute the shortest secure simple path between two nodes, source $s$ and destination $t$. In our system setting edges can have positive weight only and no loops exist in the shortest paths. In practice, while computing the shortest path routing algorithm, in general, the graph source node always picks the shortest path for routing from source $s$ and destination $t$.

Route latency can be computed by finding multiple (almost) shortest paths in the graph from source $s$ and destination $t$. Finding multiple paths is possible by generalizing the Dijkstra algorithm to find more than one path. In the literature, finding multiple shortest path problems is referred to as the $k$-shortest path problem. There are two main variations of the $k-1$ shortest path routing problem. The first is to not only find the shortest path but also $k-1$ other paths in non-decreasing shortest length (two shortest paths can have the same length) and paths are allowed to visit the same node more than once, which allows a \emph{loop}. In other variations, paths are not allowed to visit an already visited node. Paths are required to be \textit{simple} and \textit{loopless.} The $k$-shortest path routing has applications in many areas such as geographical path planning, multiple object tracking, and transportation networks.

While finding the multi-criteria $k$-shortest secure path, some nodes and edges can be traversed by more than one path which may yield accumulated weight on the shared node/edge but this can cause delay on the route. Delay can be avoided by choosing non-overlapping $k$-shortest paths, which are called $k$-disjoint shortest paths. These paths can be categorized further in two types: \emph{multi-criteria node-disjoint paths} and \emph{multi-criteria edge-disjoint paths}. These disjoint paths are distinct paths in the graph from source $s$ and destination $t$ which has a wide range of applications other than routing, such as multi-commodity flow.

\noindent
\textbf{Multi-Criteria Shortest Paths.}
Many real-life problems can be represented as a network, such as transportation networks, biological networks, and communication networks. In these networks, finding the shortest path resolves many issues such as routing and the distance between two molecules. In general, for finding the shortest path, we consider the criterion (objective) of edge weight (cost), which is called the Shortest Path Problem (SPP) with a single criterion. A Multi-Criteria Shortest Path Problem (MCSPP) consists of more than one objective while computing the shortest path between source and destination.

In the literature, many existing results are available on MCSPP. The first result on MCSPP was analyzed by Hansen \cite{Hansen1980}, which focused on a Bi-criteria Shortest Path Problem. Hansen proved that a family of problems exists, for which, any path between a given pair of nodes is a non-dominated path. Hence, any algorithm for solving MCSPP takes exponential time in the worst-case scenario. Thus, no polynomial-time algorithm can guaranteed to determine all non-dominated paths in polynomial time.

Others consider the case of Bi-criteria, introducing edges that have \emph{cost} and \emph{delay} criteria. The goal is to find an $s$-$t$ path that minimizes the cost while having a delay of at most $T$ (threshold). This scenario is known as the \emph{Restricted or Constrained Shortest Path Problem}.
Restricted Shortest Path (RSP) is often used in QoS (quality of service) routing, where the goal is to route a package along the cheapest possible
path while also satisfying some quality constraint for the user. RSP is known to be NP-Hard \cite{10.5555/574848}. As a result, many \emph{fully polynomial-time approximation scheme} (FPTAS) \cite{DBLP:conf/infocom/GoelRKL01, DBLP:journals/mor/Hassin92, DBLP:journals/orl/LorenzR01, DBLP:journals/ior/Warburton87, DBLP:conf/soda/Bernstein12, DBLP:journals/ton/XueZTT08} were designed. Heuristic based algorithms are discussed in these papers \cite{Nahrstedt,DBLP:conf/infocom/YuanL01,Krunz,DBLP:conf/infocom/MisraXY09}. 

The Pareto techniques 
are the only non-heuristics algorithms
to address the $k$-shortest multi-criteria, and they are exponential.
To the best of our knowledge, we are the first to consider specific cases and present polynomial solutions for multi-criteria $k$-shortest paths, in which we can improve the exponential Pareto solution. There are several novel approaches here, one is the ability to use lexicographical based greater than criteria comparison to support multi-criteria $k$-shortest, then proceed to embed such a lexicographic ordering using segments of bits in one integer computer word, where the bound on each segment is bounded by the sum of the criteria values it represents. The reduction of all-criteria $k$-shortest paths to a flow problem can be very useful in future research too.

\remove{
observation that summing all weights, overall edges in the graph, of particular criteria may add only $O(\log n)$ bits, and therefore will support a polynomial solution where criteria summation along a path does not intervenes with other criteria representation when the combined criteria are designed to have enough bits for each criterion. Other results are novel as well, we were happy to find that the vast number of sophisticated approximation results can be improved in many cases by our polynomial exact solutions.
}

\noindent
\textbf{Shortest Secure Paths.} The problem of the shortest secure path with respect to multi-criteria is not well studied. In Oh et al.~\cite{DBLP:conf/iccS/OhLC06} authors proposed a mechanism to find a shortest and secure path by appending the \emph{trust weight} and \emph{distance weight} for each edge. This approach improved the security level practically but no theoretical bound existed and it is limited for the case of two criterion while our paper deals with any number of criteria.

\noindent
\textbf{$k$-Shortest Path Problem.} The problem of finding the shortest paths in an edge-weighted graph is an important and well-studied problem in computer science. Dijkstra's sequential algorithm \cite{10.1007/BF01386390} computes the shortest path to a given destination vertex from every other vertex in $O(m + n \log n)$ time. The \textit{k}-shortest paths (KSP) asks to compute a set of top \textit{k}-shortest simple paths from  vertex \textit{s} to vertex \textit{t} in a digraph. In 1971, Yen \cite{10.2307/2629312} proposed the first algorithm with the theoretical complexity of $O(kn(m+n \log n))$ for a digraph with \textit{n} vertices and \textit{m} edges. \remove{The KSP problem has numerous applications in various kinds of networks such as road and transportation networks, communications networks, social networks, etc.}

The famous algorithm for the $k$-shortest path problem was proposed by Eppstein \cite{DBLP:conf/focs/Eppstein94} which allows non-simple paths (with loops) and runs in $O(m+n \log n+kn)$ time. In the initialization phase, the algorithm uses a \emph{shortest path tree} to build a data structure that contains information about all \emph{s-t} paths and how they interrelate with each other, in time $O(m+n)$. The running time for the initialization can be reduced from $O(m+n \log n)$ to $O(m+n)$ if the shortest path tree can be computed in time $O(m+n)$. In the enumeration phase, a \emph{path graph} is constructed. The path graph is a min-heap where every path starting from the common root corresponds to a \emph{s-t} path in the original graph. If we want the output paths to be sorted by the length in increasing order then the enumeration phase requires $O(k \log k)$ time. Frederickson's heap selection algorithm \cite{DBLP:journals/iandc/Frederickson93} can be used to enumerate the paths after the initialization phase in $O(k)$ time. Other KSP algorithms are discussed in these papers \cite{DBLP:journals/networks/KatohIM82,DBLP:conf/icalp/RodittyZ05,DBLP:journals/ipl/GotthilfL09,DBLP:journals/networks/Feng14,DBLP:conf/focs/HershbergerS02,DBLP:conf/isaac/KurzM16,DBLP:journals/talg/HershbergerMS07}.

\noindent
\textbf{$k$-Disjoint Shortest Path Problem.} The $k$-disjoint shortest path problem on a graph with $k$ source-destination pairs $(s_i, t_i)$ looks for $k$ pairwise node/edge-disjoint shortest $s_i$ - $t_i$ paths. The output is prioritized: the first path should be the shortest, the second one should be the shortest conditioned by that property of the first path and by disjointness, and so on. The $k$-disjoint shortest path problem is known to be NP-complete if $k$ is part of the input.

The problem of two disjoint shortest paths was first considered by Eilam-Tzoreff \cite{DBLP:journals/dam/Eilam-Tzoreff98}. Eilam-Tzoreff provided a polynomial-time algorithm for $k$= 2, based on a dynamic programming approach for the weighted undirected vertex-disjoint case. This algorithm has a running time of $O(|V|^8)$. Later, Akhmedov \cite{DBLP:conf/csr/Akhmedov20} improved the algorithm of Eilam-Tzoreff, whose running time is $O(|V|^6)$ for the unit-length case of the 2-Disjoint Shortest Path and $O(|V|^7)$ for the weighted case of the 2-disjoint shortest path. In both cases Akhmedov \cite{DBLP:conf/csr/Akhmedov20} considered the undirected vertex disjoint shortest path.

\noindent
\textbf{Organization of the paper.} Section~\ref{s:reduction} presents the reduction of multi-criteria weight to single weight. Section~\ref{s:ssp} describes the first prioritized multi-criteria $k$-shortest path algorithm. Section~\ref{s:twoDisjoint} introduces the prioritized multi-criteria 2-disjoint shortest path algorithm. The $k$-disjoint all criteria shortest path algorithm is presented in Section~\ref{s:disjointAll}, with the analysis.

\section{Finding Prioritized Multi-criteria $k$-Shortest Paths in Polynomial Time}

In this section, we reduce the value of multiple attacks on edges into a single value, which we call a reduction from multi-criteria weight to single weight.

\subsection{Reducing Multi-criteria Weight to Single Weight}
\label{s:reduction}
\begin{figure}[H]

\centering 
\includegraphics[page=1,width=\textwidth]{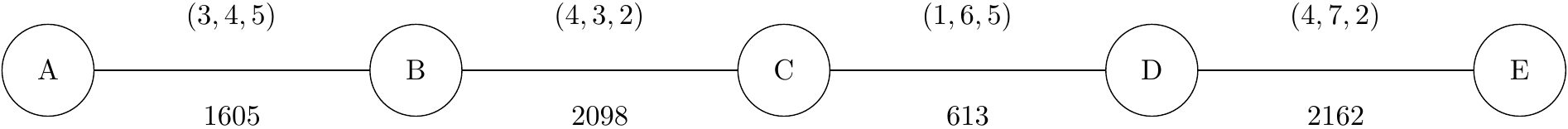}
\caption{Example for reduction of multi-criteria weight to single weight, multi-criteria weights (3,4,5) on the edge represents weight in the form of (3=$w_1$, 4=$w_2$, 5=$w_3$)}
\label{fig:fig1}
\end{figure} 
\vspace{-10mm}
\begin{figure}[H]

\centering 
\includegraphics[page=3,width=\textwidth]{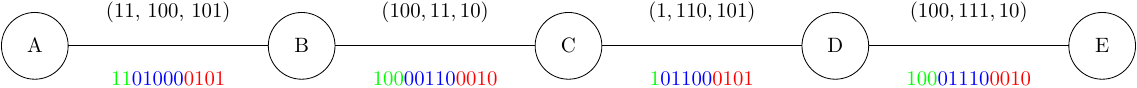}
\caption{Example for reduction of multi-criteria weight to single weight for binary bits }
\label{fig:fig2}
\end{figure} 
\vspace{-10mm}
\begin{figure}[H]

\centering 
\includegraphics[page=2,width=\textwidth]{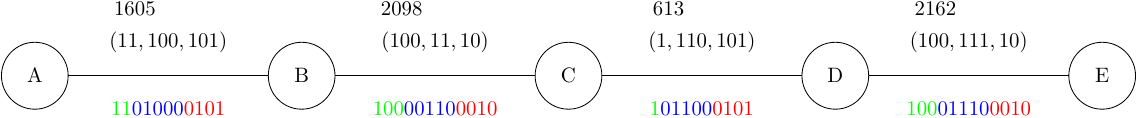}
\caption{Length of the (only and therefore) shortest path from $A$ to $E$, length of shortest path = 6478 (\textcolor{green}{1100}\textcolor{blue}{10100}\textcolor{red}{1110}), each color represent the respective sum of weights for each criterion}
\label{fig:fig3}
\end{figure}

We consider multi-criteria in weight(cost)-function in a prioritized manner. In our reduction from multi-criteria to single criterion, we ensemble the weights of the monotonic prioritized criteria into one weight. Next, we present a reduction of the prioritized multi-criteria $k$-shortest path problem to the single criterion $k$-shortest path problem. The idea is to combine the different weights into a single ``ensembled'' weight, such that the most significant part of the ensembled weight, is the weight of the most important criteria. Say, using the first most important $k_1$ bits, that suffice to accumulate the sum of weights, of the most prioritized criterion. The second most important weight resides in the next $k_2$ bits of the edge weight, and so on and so forth. Our algorithms deal with any number of attacks (criteria/weights).

We illustrate the conversion using the following example.
In the above example, we showed the reduction of multi-criteria weight to single weight. In Fig.~\ref{fig:fig1} each edge holds a vector of criteria and sum of weights. In Fig.~\ref{fig:fig2} each edge holds a vector of criteria in binary form and a sum of weights in the colored binary form where each color represents the weight of each criterion. In Fig.~\ref{fig:fig3}, we computed the only shortest path between node $A$ and $E$ and the total length of path represented in a colored binary number, in that each color represents the sum of the weight of each individual criterion. The summarized calculations are in Table~\ref{tab:example2}. 

\begin{table}
\caption{Reduction of multi-criteria weight to single weight, where $a_i$= sum of $w_i$ weights over all edges. $a_1 = 12, a_2= 20, a_3 = 14$ , $max_a = 20$, for which value of $max_a$ is ($2^5 > max_a  \geq 2^4 > 14$), so choose $2^5 $ due to 20. In fact, according to $a_3$ only four bits suffice for the description of the $w_3$ to maximal value in a shortest (any simple) path, as $14<2^4$. We can use 5 bits for the first criterion but in this example 4 bits are enough to represent $w_3$.}
\label{tab:example2}
\begin{center}

\begin{tabular}{l*{4}{r}r}
 \toprule
 & Edge 1 & Edge 2 & Edge 3 & Edge 4 \\
\hline
$w_1$ & 3 & 4  & 1 & 4 \\
$w_2$ & 4 & 3 & 6  & 7 \\
$w_3$ & 5 & 2 & 5  & 2 \\
\hline
 $w'_1 = w_1 \cdot 2^{9}$ &   1536 & 2048 & 512 & 2048 \\
 $w'_2 = w_2 \cdot 2^4$ & 64 & 48 & 96  & 112 \\
 $w'_3 = w_3 \cdot 2^{0}$ &  5 & 2 & 5 & 2  \\
\hline
\textit{Total Weight} & 1605  & 2098 & 613 & 2162 \\
\hline
\textbf{ Binary Conversion} & & &  &\\
\hline
$w_1$ & 11 & 100 & 01  & 100 \\
 $w_2$ & 100 & 11 & 110  & 111 \\
 $w_3$ & 101 & 10 & 101  & 10 \\
\hline
 $w'_1 = w_1 \cdot 2^{9}$ & 11000000000 & 100000000000 & 1000000000 &  100000000000 \\
 $w'_2 = w_2 \cdot 2^4$ & 1000000 & 110000 & 1100000  & 1110000 \\
 $w'_3 = w_3 \cdot 2^{0}$ & 101 & 10 & 101 & 10  \\
 \hline
 \textit{Total Weight} & \textcolor{green}{11}\textcolor{blue}{01000}\textcolor{red}{0101}  & \textcolor{green}{100}\textcolor{blue}{00110}\textcolor{red}{0010} & \textcolor{green}{1}\textcolor{blue}{01100}\textcolor{red}{0101} & \textcolor{green}{100}\textcolor{blue}{01110}\textcolor{red}{0010} \\
 \bottomrule
  
\end{tabular}
\end{center}
\end{table}

\subsection{Prioritized Multi-Criteria $k$-Shortest Simple Paths}
\label{s:ssp}

This section studies the following prioritized multi-criteria $k$-shortest simple paths problem. The input is an undirected graph $G=(V,E)$, where each edge $e$ holds vector $\bar w(e)$, where $\bar w(e)=(w_1(e), w_2(e), ..., w_q(e))$ and $w_i(e)$ is the weight of $e$ w.r.t.\ criterion $c_i$, source node $s$ and destination node $t$, $s,t \in V$, and integer $k$. We say that a path $P$ from $x$ to $y$ is the shortest w.r.t.\ criterion $c_i$, if $c_i(P)=\sum_{e \in P} w_i(e)$ is minimal among all $c_i(P)=\sum_{e \in P} w_i(e)$ over all paths $P$ from $x$ to $y$. A polynomial-time algorithm for solving the problem is presented in this section.

 For multiple criteria, to avoid the exponential number of paths, we reduce the set of all criteria as a single value for each edge. We reduce the prioritized multi-criteria by a reduction to a single criterion. Let us define the ensembled edge weights as follows. Let $W_i = \sum_{e \in E}^{} w_i(e)$, $1 \leq i \leq q$. Let $l_i = \ceil*{\log_2 (W_i + 1)}$, $1\leq i \leq q$, and let $r_q = 0$, $r_i = \sum_{j = i+1}^{q} l_i $, $0\leq i \leq q-1$. The ensembled weight of the edge $e\in E$ is defined to be $EW(e) =\sum_{j = 1}^{q} (2^{r_i}w_i(e))  $. As usual we define the ensembled weight of any path $P$ as $EW(P) = \sum_{e \in P}^{} EW(e)$.

The multi-criteria shortest path problem has a rich history, several approximation and heuristic-based algorithms have been proposed to solve it. Instead of considering the approximation or heuristic approach, we are interested in problem families for which a polynomial solution exists. For example, (1) if one criterion is that no edge on the path should weigh more than a given total attack threshold ($T$), then when computing the shortest multi-criteria algorithm, we do not consider this edge. (2) Another family of multi-criteria is \emph{prioritized multi-criteria} where one would like to optimize the first criterion (attack) $c_1$, and within all solutions that optimize $c_1$, find the optimal solution for the second criterion (attack) $c_2$, and so on. (3) A combination of the two multi-criteria above.

Thus, as explained above, to ensemble the weights of the monotonic prioritized criteria into one weight, we use the most important part of an edge ensemble weight for the most important criteria, and the least important part of an edge ensemble weight for the least important criteria, and similarly for criteria in between. 

\begin{algorithm}[!t]
\caption{Generalized Dijkstra Algorithm for Multi-Criteria Shortest Path}
\label{algo:Dijkstra}
\KwIn{An undirected graph $G=(V,E)$, where each edge $e$ holds vector $\bar w(e)$, where $\bar w(e)=(w_1(e), w_2(e), ..., w_q(e))$ and $w_i(e)$ is the weight of $e$ w.r.t.\ criterion $c_i$, source node $s$, destination node $t$, and threshold $T$}

\KwOut{Multi-criteria shortest simple path}

 \SetKwFunction{FMain}{}
 \SetKwProg{Fn}{Procedure}{ GeneralizedDijkstra(Graph $G$, Source $s$, Destination $t$)}{}
 \DontPrintSemicolon
 \Fn{}{
    Initialize: Source \{s\}\\
    $dist(s)$ = 0\\
    $EW$ = 0\\
    $r_q$ = 0\\
    
        \SetKwFunction{FMain}{}
        \SetKwProg{Fn}{for each vertex $v$ except for $s$}{ do}{}
        \DontPrintSemicolon
        \Fn{}{
                $dist(v)$ =  $\infty$
        }
    
    $S= \phi$\\
    $Q = V$\\
    Compute $W_i$ = $\sum_{e \in E} w_i(e)$, $1 \leq i \leq q$\\
    Compute $l_i$= $\lceil \log_2 (W_i + 1) \rceil$, $1 \leq i \leq q$\\
    Compute $r_i$ = $\sum_{j = i+1}^{q} l_i$, $0 \leq i \leq q-1$  \\
    
    \SetKwFunction{FMain}{}
             \SetKwProg{Fn}{for each edge $e$}{}{}
             \DontPrintSemicolon
             \Fn{}{
                   $EW(e)$ = $\sum_{j = 1}^{q} (2^{r_i}w_i(e))$ \\
                    \If{$EW(u,v) \geq T$}{
                    Delete $edge(u,v)$
                        }
                    }
    
    \While{$Q \neq \phi$}{
        $u=$ Extract-Min($Q$)\\
        $S = S \cup \{u\}$

        \SetKwFunction{FMain}{}
        \SetKwProg{Fn}{for each vertex $v \in G.Adj[u]$}{ do}{}
        \DontPrintSemicolon
        \Fn{}{

                $Relax (u,v)$
              
        }
     }
 }

 \SetKwFunction{FMain}{}
 \SetKwProg{Fn}{Procedure}{ Relax($u,v$)}{}
 \DontPrintSemicolon
 \Fn{}{
             \SetKwFunction{FMain}{}
             \SetKwProg{Fn}{for each neighbor of $u$}{}{}
             \DontPrintSemicolon
             \Fn{}{

                    \If{$dist(v) > dist(u) + EW(u,v)$ }{
                        $dist(v) = dist(u) + EW(u,v)$
                    }
                    
             }
 
 }

\end{algorithm}

To make sure that the portion of edge weight dedicated to criteria does not overlap, we assign each portion a span of bits in the ensemble weight of an edge to suffice for accumulating the criteria weight along the (shortest) path. We can bound the number of bits needed for accumulating the bound on the shortest path by summing up all weights of the criteria in all edges in the graph. 

Finding the $k$-shortest paths with the ensemble weights results in that these are $k$-shortest paths in the most important criterion $c_1$, as all other criteria do not compete with the most important part of the weights when computing the shortest path(s). Thus, the second criterion $c_2$ breaks ties among the paths as above with the same value of the first criterion. In particular, if the weight of the heaviest shortest path according to $c_1$ is $w_1$, the selection from the set of the shortest paths with weight $w_1$ will be according to the second prioritized criterion $c_2$. If the set of shortest paths with $w_1$ is chosen according to the second criterion where $w_2$ is the shortest among them, then from the set of paths with weights $w_1$ and $w_2$, paths with the lightest weight according to the third criterion are chosen, and so on and so forth. 

\remove{
\begin{algorithm}
\caption{Multi-Criteria $k$-Shortest Secure Simple Paths}
\label{algo:MCKSSP}
\KwIn{An undirected graph $G= (V,E)$, source node $s$, destination node $t$, and number of shortest path to find $k$
}
\KwOut{Multi-criteria $k$-shortest secure simple paths}

$\cal P$ = Empty path set\\
$P$ = Empty stack\\
\emph{noMorePaths} = $false$

 \SetKwFunction{FMain}{}
 \SetKwProg{Fn}{while(|$\cal P$| < $k$ $\land$ \emph{noMorePaths} = $false$ ) }{ do}{}
 \DontPrintSemicolon
 \Fn{}{
        $G \leftarrow reduce(G)$\\
        $P$ = $GeneralizedDijkstra(G, s, t)$\\
        \If{|$P$| = 0}{
            \emph{noMorePaths} = $true$ 
        }
        \Else{
            $\cal P$ = $\cal P$ $ \cup$ {$P$}\\
            $E \leftarrow$ $E \backslash \{ lighestEdge(P)\}$
        }
 
 }

return $\cal P$
\end{algorithm}
}

The ensemble of the criteria weight into one weight implies finding monotonic multi-criteria $k$-shortest paths. These paths are not necessarily disjoint (as the $k$-shortest simple paths) and also for edge/node disjoint paths. These paths can be computed in polynomial time as long as $k$ is fixed~\cite{DBLP:journals/dam/Eilam-Tzoreff98}.

Our approach is based on the generalized Dijkstra algorithm~\cite{DBLP:journals/corr/abs-2101-11514} for the multi-criteria shortest path. Using the Dijkstra algorithm, it is possible to determine the shortest distance (or the least attack value) between a start node and any other node in a graph. The idea of the algorithm is to continuously apply the original Dijkstra algorithm with the precomputed ensembled weight for each edge while removing the edges that hold more ensembled weight than Threshold ($T$).

Our approach consists of the following steps: $Q$ is the set of nodes for which the shortest path has not been found. Initialize the source node with distance 0 and all nodes with distance ``infinite''. Reduce the multi-criteria into a single criterion. At each iteration, the node $v$ that has the minimum distance value (sum of weights $EW$) to the source is added to the $S$, which provides the shortest path from the source node to the destination node.

For computing the multi-criteria $k$-shortest simple paths, we use the Algorithm~\ref{algo:Dijkstra} (generalized Dijkstra algorithm) for computing the shortest path and Yen's algorithm \cite{10.2307/2629312} for computing $k$-paths\footnote{Finding $k$-shortest path can be done using different approaches. Such as by removing the lightest edge~\cite{Victor2017} and by removing the already found shortest path(s)~\cite{5473995}. Our algorithm and reduction from multi-weight to single weight work with both scenarios.}.  

\remove{
\textcolor{red}{Then take $(k-1)^{th}$ shortest path and make each node in the path unreachable in turn, i.e. remove a particular edge that goes to the node within the route. Once the node is unreachable, find the shortest path from the preceding node to the destination. Then we have a new path that is created by appending the common sub-path (from the source node to the preceding node of the unreachable node) and adds the new shortest path from the preceding node to the destination node. This path is added to the list $B$, only when it has not appeared in list $A$ or list $B$ before. After repeating this for all nodes in the path, we have to find the shortest path in list $B$ and move that to list $A$. We just have to repeat this process for $k$ times.}}

\remove{
\textcolor{red}{\textbf{Remark: }
Finding $k$-shortest path can be done using different approaches. Such as by removing the lightest edge~\cite{Victor2017} and by removing the already found shortest path(s)~\cite{5473995}. Our algorithm and reduction from multi-weight to single weight work with both scenarios.}
}

\begin{theorem}
The prioritized multi-criteria $k$-shortest paths problem in an undirected graph can be solved in polynomial time.
\end{theorem}

\begin{proof}
The single criterion $k$-shortest paths problem is solvable in polynomial time. We polynomially reduced the multi-criteria weights where criteria are used in a prioritized manner to the single criterion weight. So prioritized multi-criteria $k$-shortest paths problem is also solvable in polynomial time.  \qed
\end{proof}


\section{Prioritized Multi-Criteria 2-Disjoint (Node/Edge) Shortest Paths} 
\label{s:twoDisjoint}
In this section, we suggest an algorithm solving the 2-shortest paths edge/node independent problem (see Eilam-Tzoreff \cite{DBLP:journals/dam/Eilam-Tzoreff98}) for the case of prioritized criteria from a single source $s$ to a single destination $t$ in an undirected graph $G$, where each edge $e$ holds vector $\bar w(e)$, where $\bar w(e)=(w_1(e), w_2(e), ..., w_q(e))$ and $w_i(e)$ is the weight of $e$ w.r.t.\ criterion $c_i$, source node $s$ and destination node $t$, $s,t \in V$. A polynomial time algorithm solving the problem is presented in this section.

We reduce the prioritized multi-criteria case to the case of a single criterion, similarly to Section \ref{s:reduction}.
Further, for finding 2-disjoint shortest paths from $s$ to $t$, we use a reduction to the case where two sources and two destinations are given (described later). 
Then, we find the 2-disjoint shortest paths in the resulted graph $\bar G$ by using the algorithm of Akhmedov \cite{DBLP:conf/csr/Akhmedov20}, which
computes the 2-disjoint shortest paths for two sources and two destinations in time $O(|V|^7)$.

Let us describe our reduction.
For the edge-disjoint case, it is simple. We add to $G$ two nodes $s_1, s_2$ with dummy edges $(s_1,s), (s_2,s)$, two nodes $t_1, t_2$ with dummy edges $(t,t_1), (t,t_2)$, define the weight to be zero for the dummy edges, and declare $s_1, s_2$ to be the sources and $t_1, t_2$ to be the destinations instead of $s$ and $t$. After finding the 2-disjoint shortest paths in the resulting graph $\bar G$, we return them with the dummy edges removed.

The reduction for the node-disjoint case is more complicated. We add to $G$ four nodes $s_1, s_2, t_1, t_2$, which will be the sources and destinations instead of $s,t$. If $G$ contains edge $(s,t)$, then we replace it with edge $(s',t')$ of the same weight, and add the dummy edges $(s_1,s')$, $(s_2,s')$, $(t',t_1)$, $(t',t_2)$. 
For any other edge $(s,v)$ incident to $s$, we replace it with edge $(s_v,v)$ of the same weight, where $s_v$ is a new node, and add the dummy edges $(s_1,s_v), (s_2,s_v)$. Symmetrically, for any other edge $(v,t)$ incident to $t$, we replace it with edge $(v,t_v)$ of the same weight, where $t_v$ is a new node, and add the dummy edges $(t_v,t_1), (t_v,t_2)$.
The weights of all dummy  edges are set be 1.
Finally, we remove nodes $s$ and $t$ with no incident edges.
After finding the 2-disjoint shortest paths in the resulting graph $\bar G$, we return their \emph{abridged variant}: with the dummy edges incident to their end-nodes $s_i$ and $t_i$ shrunken to $s$ and $t$, respectively.

Let us show the correctness of the latter reduction.
A necessary condition for using the algorithm of Akhmedov is that the terminals quadruple $(s_1, s_2, t_1, t_2)$ is not rigid, where it is called rigid, is if $s_1, t_1 \in L(s_2,t_2)$ and $s_2, t_2 \in L(s_1,t_1)$, where $L(s_i, t_i)$ is the set of all nodes belonging to at least one shortest path between $s_i$ and $t_i$. Let us prove that $(s_1, s_2, t_1, t_2)$ is not rigid in $\bar G$. Assume for the contradiction, w.l.o.g., that a shortest path $P$ from $s_1$ to $t_1$ in $\bar G$ contains two consequent edges $(u,s_2)$ and $(s_2,v)$. Consider path $P'$ obtained from $P$ by replacing its prefix from $s_1$ to $v$ by edge $(s_1,v)$, with weight 1. Path $P'$ from $s_1$ to $t_1$ is lighter than $P$, since the weights of  the three removed edges: the first one of $P$, $(u,s_2)$, and $(s_2,v)$, are 1 each,---a contradiction.

\begin{algorithm}

\caption{Prioritized Multi-criteria 2-Disjoint Shortest Paths}
\label{algo:disjoint}
\KwIn{An undirected graph $G=(V,E)$, where each edge $e$ holds vector $\bar w(e)$, where $\bar w(e)=(w_1(e), w_2(e), ..., w_q(e))$ and $w_i(e)$ is the weight of $e$ w.r.t.\ criterion $c_i$, source node $s$ and destination node $t$ }

\KwOut{Prioritized multi-criteria 2-disjoint shortest paths from $s$ to $t$}

    Compute $W_i$ = $\sum_{e \in E} w_i(e)$, $1 \leq i \leq q$\\
    Compute $l_i$= $\lceil \log_2 (W_i + 1) \rceil$, $1 \leq i \leq q$\\
    Compute $r_i$ = $\sum_{j = i+1}^{q} l_i$, $0 \leq i \leq q-1$  \\
    
        \SetKwFunction{FMain}{}
             \SetKwProg{Fn}{for each edge $e \in E$}{}{}
             \DontPrintSemicolon
             \Fn{}{
                        $EW(e)$ = $\sum_{j = 1}^{q} (2^{r_i} \cdot w_i(e))$ \\
             }

    Construct extended graph $\bar G$ with two sources $s_1, s_2$ and two destinations $t_1, t_2$ as described in the text \\
    $(\bar P_1, \bar P_2) \leftarrow 2DSP(\bar G,EW,s_1, t_1, s_2, t_2)$  \\
    transform $\bar P_1, \bar P_2$ to to their abridged variants $P_1, P_2$ (see the text)\\
    return $(P_1, P_2)$

    \SetKwFunction{FMain}{}
    \SetKwProg{Fn}{Procedure}{ 2DSP($G',w,s_1, t_1, s_2, t_2$)}{}
    \DontPrintSemicolon
    \Fn{}{
            Execute the algorithm of Akhmedov \cite{DBLP:conf/csr/Akhmedov20} for computing 2-disjoint shortest paths in graph $G'$ with edges weights $w$
     }
             
\end{algorithm}

Let us show the legality and optimality of the returned solution.
Since the paths returned by the algorithm of Akhmedov are node-disjoint, their abridged variants are node-disjoint. 
By the reasons in the proof of non-rigidity of $(s_1, s_2, t_1, t_2)$, the returned paths 
do not contain any terminal out of $s_1, s_2, t_1, t_2$ as an intermediate node. Therefore, their abridged variants do not contain dummy edges, and thus are legal paths in $G$.
Let $(P^*_1, P^*_2)$ be the optimal pair of 2-disjoint shortest paths from $s$ to $t$ in $G$. The paths corresponding to them in $\bar G$---obtained from them by the operations in the reduction applied to their first and last edges---are node-disjoint and have weights greater by 2 than the weights of  $P^*_1$ and $P^*_2$. The optimal paths in $\bar G$ are not worse, and their abridged variants are by lighter by 2. Therefore, the pair of paths returned by the reduction is not worse than the pair $(P^*_1,P^*_2)$, as required.

   \begin{figure*}
        \centering 
        \includegraphics[width=0.80\textwidth]{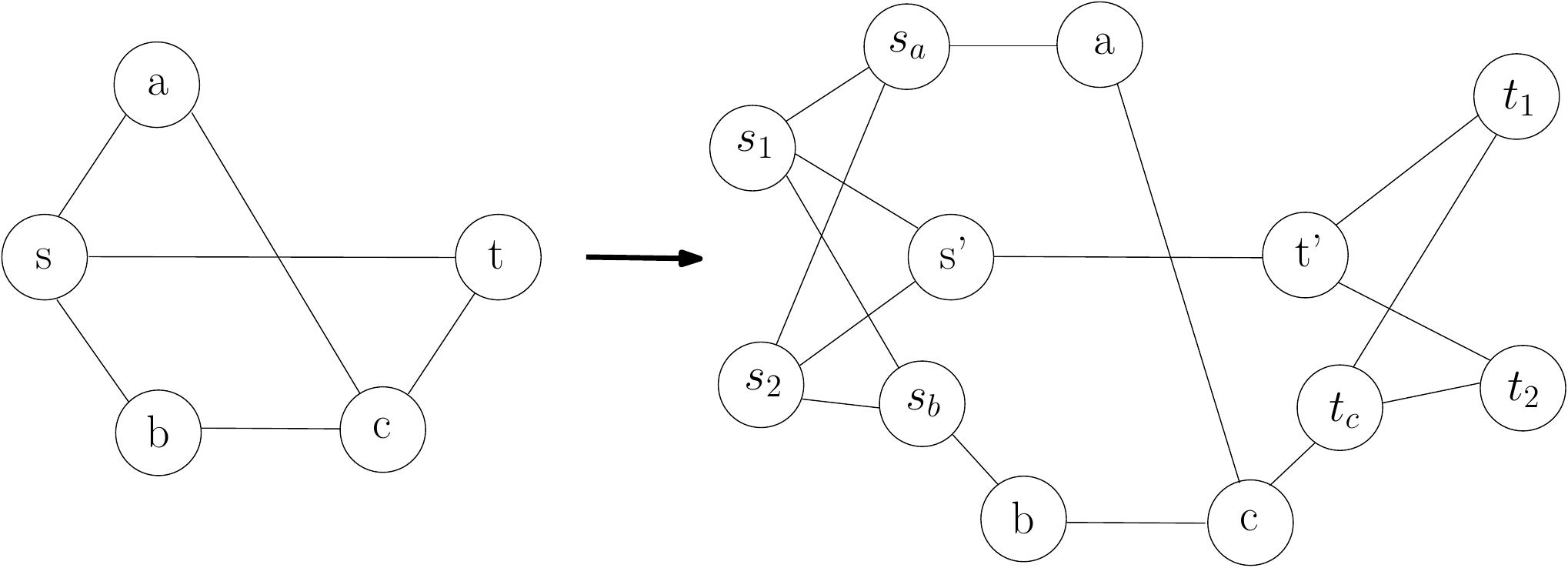}
        \caption{Reduction for the node-disjoint case }
        \label{fig:fig4}
        \end{figure*}

The algorithm as above is presented in pseudo-code in Algorithm \ref{algo:disjoint}. Its correctness, together with the polynomiality of the algorithm of Akhmedov and of the reduction, implies the following statement.

\begin{theorem}
The prioritized multi-criteria 2-disjoint shortest paths problem in an undirected graph can be solved in polynomial time.
\end{theorem}

We illustrate the reduction for the node-disjoint case in example (Fig. \ref{fig:fig4}).

\section{$k$-Disjoint All-Criteria-Shortest Paths}
\label{s:disjointAll}

This section studies the following $k$-disjoint all-criteria-shortest paths problem. The input is a directed graph $G=(V,E)$, $q$ weight functions 
$w_i$ on edge set $E$, $1 \le i \le q$, source node $s$ and destination node $t$, $s,t \in V$, and integer $k$. We say that a path $P^*$ from $x$ to $y$ 
is the shortest w.r.t.\ criterion $c_i$, if $c_i(P^*)=\sum_{e \in P^*} w_i(e)$ is minimal among all $c_i(P)=\sum_{e \in P} w_i(e)$ over all paths $P$ from $x$ to $y$. A set of $k$ (edge-)disjoint paths from $s$ to $t$ such that each one of them is \emph{shortest regarding each one of the $q$ criteria} is sought for, if exists.
After some analysis, a polynomial algorithm solving the problem is presented and analyzed. 
We first reduce the problem to its single criterion version, then reduce the latter problem to finding $k$ disjoint paths from $s$ to $t$ in a certain sub-graph of $G$, if they exist. Finally we present an algorithm for finding them, if they exist, using known techniques: max-flow finding and flow decomposition.

We assume that each node is reachable from $s$ and that $t$ is reachable from each node in $G$; otherwise, the extra nodes could be removed from $G$.
Let distances $d_i(s,x)$ and $d_i(y,t)$, $x,y \in V$, denote the lengths of the shortest paths from $s$ to $x$ and from $y$ to $t$, respectively, w.r.t.\ criterion $c_i$ in $G$.
We assume that there is no negative cycle in $G$ w.r.t.\ any weight function $w_i$, in order for the shortest paths to exist. (See, e.g., \cite{Corman09} for the basic information on graph algorithms.)

Let us define the auxiliary \emph{aggregated weight} $w(e) = \sum_i w_i(e)$ for each edge $e \in E$, and the auxiliary \emph{aggregated criterion} $c(P)=\sum_{e \in P} w(e)$ for each path $P$ in $G$. Note that there is no negative cycle in $G$ w.r.t.\ weight function $w$, by our assumption; hence, shortest paths w.r.t.\ $c$ exist, and thus distances $d(s,x)$ and $d(y,t)$, $x,y \in V$, w.r.t.\ $w$ are well defined. 
Observe that $d(s,t) = \min_{P} c(P) = \min_{P} \sum_i c_i(P) \ge \sum_i \min_{P} c_i(P) = \sum_i d_i (s,t)$, where each minimum is taken over all paths $P$ from $s$ to $t$ in $G$.
Moreover, the equality $d(s,t) = \sum_i d_i (s,t)$ holds if and only if there exist paths from $s$ to $t$ shortest  w.r.t.\ each criterion $c_i$; in this case, these paths and only these are shortest w.r.t.\ criterion $c$.

As a consequence, we obtain the reduction from our problem to the auxiliary single criterion problem, as follows. If $d(s,t) > \sum_i d_i (s,t)$, then no paths from $s$ to $t$ shortest  w.r.t.\ each criterion $c_i$ exist. Checking this could be made via $q+1$ executions of algorithm Dijkstra on $G$, w.r.t.\ each criterion $c_i$ and w.r.t.\ criterion $c$. Otherwise, the required $k$ disjoint paths are the \emph{$k$-disjoint paths shortest w.r.t.\ criterion $c$, if they exist}. In what follows, we present the solution---an analysis and an algorithm---to the single criterion disjoint shortest $k$ paths problem. 
Let us begin with the problem analysis.

\begin{lemma}
\label{l:sub-graph}
\begin{enumerate}
	\item
	Node $u$ belongs to at least one shortest path from $s$ to $t$ if and only if $d(s,u) + d(u,t) = d(s,t)$.
		\item
	Edge $(u,v)$ belongs to at least one shortest path from $s$ to $t$ if and only if $d(s,u) + w(u,v) + d(v,t) = d(s,t)$.
\end{enumerate}
\end{lemma}

\begin{proof}
	(1) If a path $P$ from $s$ to $t$ going via $u$ is shortest, then it is known that its parts from $s$ to $u$ and from $u$ to $t$ are also shortest. In other words, their lengths are $d(s,u)$ and $d(u,t)$, respectively. The equation as required follows. If $d(s,u) + d(u,t) = d(s,t)$, then the concatenation of the shortest paths from $s$ to $u$ and from $u$ to $t$ is a path from $s$ to $t$ of length $d(s,t)$. The proof of item (2) is similar. \qed
\end{proof}

Let us define $\tilde V$ as the subset of nodes as in Lemma~\ref{l:sub-graph}(1) and $\tilde E$ as the subset of edges as in Lemma~\ref{l:sub-graph}(2). We denote by $\tilde G$ the (sub-)graph $(\tilde V,\tilde E)$.

\begin{lemma}
	\label{l:any_path}
	\begin{enumerate}
		\item
		Each shortest path from $s$ to $t$ is contained in $\tilde G$.
		\item
		Each path from $s$ to $t$ contained in $\tilde G$ is shortest.
	\end{enumerate}
\end{lemma}

\begin{proof}
	The proof of item (1) is a straightforward corollary from Lemma~\ref{l:sub-graph}. 
	
	(2) Let $P$ be any path from $s$ to $t$ in $\tilde G$. Assume to the contrary that $c(P) > d(s,t)$. Let us denote by $P_v$ the prefix of $P$ ending at $v$, $v \in P$. Note that the (degenerate) path $P_s$ from $s$ to itself is shortest: $c(P_s) = 0 = d(s,s)$. Let $v$ be the first node on $P$ such that $c(P_v) > d(s,v)$; let $P'$ be a shortest path from $s$ to $v$, $c(P') < c(P_v)$. Denote by $P''$ some shortest path from $v$ to $t$. Let $(u,v) \in \tilde G$ be the edge on $P$ entering $v$. By definition of $v$, $c(P_u) = d(s,u)$. By Lemma~\ref{l:sub-graph}(2), $d(s,u) + w(u,v) + d(v,t) = d(s,t)$. Let us concatenate $P'$ and $P''$.
	$$c(P' \cdot P'') < c(P_v \cdot P'') =  c(P_u \cdot (u,v) \cdot P'') = d(s,u) + w(u,v) + d(v,t) = d(s,t).$$
	Thus, $c(P' \cdot P'') < d(s,t)$,---a contradiction to the definition of $d(s,t)$. \qed
\end{proof}

%
%

By Lemma~\ref{l:any_path}, we have a reduction from the single criterion disjoint shortest $k$ paths problem to \emph{finding $k$ disjoint paths from $s$ to $t$ in $\tilde G$, if they exist}. Finding such paths, if they exist, may be done by known max-flow techniques. 
Let $N$ be the flow network ($\tilde G,s,t$) with unit capacities of all its edges. In what follows, we omit detailed proofs, since the material is basic.

\begin{algorithm}[H]
\caption{$k$-Disjoint All-Criteria-Shortest Paths Finding}
\label{algo:allcriteria}
\KwIn{A directed graph $G= (V,E)$, weight functions $w_i$, $1 \le i \le q$, on edge set $E$, source node $s$, destination node $t$, and integer $k$}
\KwOut{$k$-Disjoint All-Criteria-Shortest Paths, if they exist}

    Compute the aggregated weight function $w(e) = \sum_{i=1}^{q} w_i(e)$, for all edges $e \in E$\\
    Run $q+1$ times algorithm Dijkstra for finding distance functions $d$ and $d_i$ w.r.t.\ weights $w$ and $w_i$, respectively, $1 \le i \le q$ \\
    \If{$d(s,t) > \sum_{i=1}^{q} d_i(s,t)$}{
        return ``No path from $s$ to $t$ shortest w.r.t.\ each criterion $c_i$ exist''
    }
    \Else{
        set sets $\tilde V$ and $\tilde E$ be empty\\       
        \SetKwFunction{FMain}{}
        \SetKwProg{Fn}{for each node $u$ in $V$}{}{}
        \DontPrintSemicolon
        \Fn{}{
            \If{$d(s,u) + d(u,t) = d(s,t)$}{
                    add node $u$ to $\tilde V$
            }   
        }
        \SetKwFunction{FMain}{}
        \SetKwProg{Fn}{for each edge $(u,v)$ in $E$}{}{}
        \DontPrintSemicolon
        \Fn{}{
                  \If{$d(s,u) + w(u,v) + d(v,t) = d(s,t)$}{
                    add edge $(u,v)$ to $\tilde E$
                }  
            }
        construct flow network $N = (\tilde G= (\tilde V, \tilde E), s,t)$ with unit capacities on all edges\\
        run the max-flow algorithm on $N$, finding flow $f^*$ \\
        \If{the value of $f^*$ is less than $k$}{
        return ``There exist no $k$ paths from $s$ to $t$ shortest w.r.t.\ each criterion $c_i$''
        }
        \Else{
        
        initialize path set $\cal P$ be empty\\ 
           
        \SetKwFunction{FMain}{}
        \SetKwProg{Fn}{repeat $k$ times}{}{}
        \DontPrintSemicolon
        \Fn{}{  
            \SetKwFunction{FMain}{}
            \SetKwProg{Fn}{Phase 1}{}{}
            \DontPrintSemicolon
            \Fn{}{
                set stack $S$ be empty and set $v$ be $t$, and mark it\\
                \SetKwFunction{FMain}{}
                \SetKwProg{Fn}{repeat}{}{}
                \DontPrintSemicolon
                \Fn{}{
                    choose an edge $(u,v)$ with $f(u,v)=1$ and push it into $S$\\
                    set $v$ be $u$, and mark it \\
                    If $v=s$, break the repeat loop and go to Phase 3\\
                    If $v$ is marked, suspend the repeat loop and go to Phase 2\\
                     }  
                }
            \SetKwFunction{FMain}{}
            \SetKwProg{Fn}{Phase 2}{}{}
            \DontPrintSemicolon
            \Fn{}{
                set $z$ be $v$\\
                \SetKwFunction{FMain}{}
                \SetKwProg{Fn}{repeat}{}{}
                \DontPrintSemicolon
                \Fn{}{
                    pop edge $(u,v)$ from $S$, and set $f(u,v)$= 0\\
                    set $v$ be $u$, and unmark it \\
                    If $v=z$, mark $v$ and resume the repeat loop of Phase 1\\
                     }         
                }
                \SetKwFunction{FMain}{}
                \SetKwProg{Fn}{Phase 3}{}{}
                \DontPrintSemicolon
                \Fn{}{
                    set edge list $P$ be empty and
                    unmark $v=s$\\
                    \SetKwFunction{FMain}{}
                    \SetKwProg{Fn}{repeat while $v \neq t$}{}{}
                    \DontPrintSemicolon
                    \Fn{}{
                         pop edge $(v,u)$ from $S$\\
                         set $f(u,v)$= 0, and add $(v,u)$ to $P$\\
                         set $v$ be $u$, and unmark it \\
                        }
                    add $P$ to $\cal P$\\
                    }  
        }
    return $\cal P$
    }
}
\end{algorithm}

\begin{proposition}
\label{p:prop1}
	A set of $k$ disjoint paths from $s$ to $t$ in $\tilde G$ exists if and only if the value of maximal flow in $N$ is at least $k$.
\end{proposition}

\begin{proof}
   \textbf{direction only if}
   assume that $\cal P$ is a set of $k$ disjoint paths from $s$ to $t$. Let us define flow $f$ by setting it to 1 on all edges belonging to the paths in $\cal P$ and to 0 on all other edges. It is easy to see that $f$ is a flow of value $k$ from $s$ to $t$ in $N$. Hence, the value of a max-flow in $N$ is at least $k$.
   
   \textbf{direction if}
   let $f$ be an integer (that is 0/1) flow in $N$ of value at least $k$. 
   A set $\cal P$ of $k$ disjoint paths from $s$ to $t$ in $\tilde G$ is produced by executing the following triple-phased \emph{path-finding routine} $k$ times, beginning from $f=f_0$ and an empty path set $\cal P$. 
   We denote by $E_f$ the (dynamic) set of edges in $N$ with the value 1 of flow $f$.

   \emph{Phase 1}
   Set stack $S$ to be empty. Set to $v$ be $t$ and mark it. Choose an edge $(u,v)$ in $E_f$ and push it into $S$. Set $v$ to be $u$, mark it, and continue in the same way. If we arrived at $v=s$, go to Phase 3. If we arrived at a marked node $v$, suspend Phase 1 and go to Phase 2.
   
   \emph{Phase 2}
   Set $z=v$. Pop edge $(u,v)$ from $S$. Unmark $u$, set $f(u,v)=0$, set $v$ to be $u$, and continue in the same way. When arrived at $v=z$, mark $v$ and resume Phase 1.
   
   \emph{Phase 3}
   Set list $P$ to be empty. Unmark $v=s$. Pop edge $(u,v)$ from $S$. Unmark $u$, add $(u,v)$ to $P$, set $f(u,v)=0$, set $v$ to be $u$, and continue in the same way. Upon arrival at $v=t$, add $P$ to $\cal P$.

Let us explain briefly the correctness. After each execution of Phase 2 (removal of a flow cycle), $f$ becomes a correct flow with the same value. After each execution of Phase 3 (removal of a flow path), $f$ becomes a correct flow with a value smaller by 1. 
At Phase 1, since the flow value is non-zero, there exists at least one edge in $E_f$ entering $t$. Since flow $f$ is correct, at each node $v$, the numbers of edges incoming and outgoing $v$ are equal in $E_f$. Hence, if there is an edge outgoing $v$ in $E_f$, then also an edge incoming $v$ exists in $E_f$. \qed
\end{proof}

We conclude that for solving the problem, it is sufficient to find a max-flow in $N$, and if the flow value is at least $k$, execute the path-finding routine $k$ times.  The pseudo-code of the described solution scheme is presented in Algorithm~\ref{algo:allcriteria}.
Let us analyse the running time. A flow either of value $k$, if such a one exists, or a max-flow, otherwise, in a network with unit edge capacities can be found in time $O(\min(k|E|, |E|^{3/2}))$ (see \cite{ahuja1993network, even_2011}). 
All executions of the path-finding routine together take $O(|E|)$ time, since each edge is processed in time $O(1)$ in total. Summarizing, the running time bound is $O(\min(k|E|, |E|^{3/2}))$. 


\section{Conclusion}
We presented polynomial time multi-criteria (secure) paths algorithms, which include prioritized multi-criteria $k$-shortest (secure) paths algorithm, multi-criteria 2-disjoint (vertex/edge) shortest (secure) paths algorithm for a undirected graph, and $k$-disjoint all-criteria-shortest (secure) paths algorithm for a directed graph. We believe that we open a ground for exploring more cases in which polynomial multi-criteria $k$-shortest paths may exist.

%
%
%
\bibliographystyle{splncs04}
\bibliography{mybibliography}

\end{document}